\documentclass{article}

\pdfoutput=1
\usepackage{url}

\usepackage{natbib}

\usepackage[utf8]{inputenc}
\usepackage{amssymb}
\usepackage{amsmath}
\usepackage{amsthm}
\usepackage{algorithm}
\usepackage{authblk}
\usepackage[noend]{algpseudocode}
\usepackage{float}
\usepackage{graphicx}
\usepackage{xcolor}
\usepackage{soul}
\DeclareUnicodeCharacter{207B}{-}
\newtheorem{theorem}{Theorem}
\newtheorem{corollary}[theorem]{Corollary}
\newtheorem{definition}[theorem]{Definition}
\newtheorem{problem}[theorem]{Problem}

\allowdisplaybreaks

\title{\textbf{Balancing the Payment System}}
\author[1]{Tomaž Fleischman}
\author[2]{Paolo Dini}
\affil[1]{Be Solutions d.o.o.,
Bleiweisova cesta 30, 1000 Ljubljana, Slovenia\authorcr
e-mail: tomaz.fleischman@be-solutions.si
\vspace{.2cm}}
\affil[2]{London School of Economics and Political Science, UK\authorcr
e-mail: p.dini@lse.ac.uk}

\date{October 2020}

\begin{document}

\maketitle

\begin{abstract}
\noindent
The increasingly complex economic and financial environment in which we live makes the management of liquidity in payment systems and the economy in general a persistent challenge. New technologies are making it possible to address this challenge through alternative solutions that complement and strengthen existing payment systems. For example, the interbank balancing method can also be applied to private payment systems, complementary currencies, and trade credit clearing systems to provide better liquidity and risk management. In this paper we introduce the concept of a balanced payment system and demonstrate the effects of balancing on a small example. We show how to construct a balanced payment subsystem that can be settled in full and, therefore, that can be removed from the payment system to achieve liquidity-saving and payments gridlock resolution.
We also briefly introduce a generalization of a payment system and of the method to balance it in the form of a specific application (Tetris Core Technologies), whose wider adoption could contribute to the financial stability of and better management of liquidity and risk for the whole economy.
\end{abstract}

\section{Introduction}

A \emph{payment system} consists of two main parts, the obligation network and the liquidity source/sink, which is used to facilitate the discharge of obligations in the obligation network. A \emph{balanced payment system} is a payment system in which all obligations can be discharged simultaneously. This is possible when the total inflow of cash equals the total outflow of cash for every part of the system -- meaning that the system satisfies the flow conservation condition, such that the conclusion that this clears all obligations in the system is trivial.

Constructing a balanced system has practical value. For example, subtracting a balanced subsystem does not disrupt the balance of the remaining payment system. This means that the subtraction of a balanced subsystem will decrease the total debt in a payment system without disturbing the relative liquidity positions between the remaining counterparties with non-zero net positions. The key to constructing a balanced system at any one time is a centralized knowledge of the obligations that are present at that time between the members of the obligation network. This centralised knowledge allows to maximise the amount of mutual indebtedness that can be taken out of the obligation network. The benefit of membership is liquidity-saving for the participating members and a decreased systemic risk for the members and the wider economy.

An obligation network can be viewed as a set of payments due. Payments reflect the complex and highly interconnected supply networks and form a dense strongly-connected obligation network. `Strongly connected' means that there is a path of payments or invoices in each direction connecting any pair of firms.\footnote{Such a path usually involves multiple, and different, firms in each direction.} If the obligation network is not strongly connected then it can usually be split into just a few strongly-connected parts or ``clusters''. A consequence of this definition is that all the firms in a strongly-connected network are part of at least one cycle. Although this sounds encouraging, depending on the distribution of liquidity over the payment system members we can observe situations where payments cannot be processed individually. Leinonen \cite{leinonen2005liquidity} provides the following definitions for different possible liquidity distributions:
\begin{itemize}
    \item \textbf{Circular} - is a situation where individual payments can only be settled in a specific order. This situation is resolvable by reordering the payment queue.
    \item \textbf{Gridlock} - is a situation in which several payments cannot be settled individually but can be settled simultaneously. This situation is resolvable with multilateral off-set.
    \item \textbf{Deadlock} - is a situation where the individual payments can be made only by adding liquidity to at least one of the system participants.
\end{itemize}

The methods used to resolve these situations are called Liquidity-Saving Mechanisms (LSMs). The benefits of LSMs in interbank payment systems are well described and demonstrated in \cite{RePEc:bof:bofrdp:2001_009} on a set of real data. An LSM applied to a payment system shortens the queues and reduces the need for additional liquidity to discharge the obligations.

The use of LSMs in interbank payment systems is widespread, but the benefits of liquidity-saving do not reach everyone. In particular, small companies with limited access to liquidity often use various alternative ways to discharge their obligations on the trade credit market or through the use of complementary currencies. This brings up the question of systems interdependence and risk of liquidity problems spilling from one system to another. Foote \cite{ASYMandSPILL} shows that the use of an LSM in one system reduces this risk in all systems.

Similarly to a payment system, a \emph{balanced obligation network} is a network that can discharge all obligations simultaneously. The most important task of any LSM is to find such a network. To manage the risks in the context of the increasing complexity of payment systems, the concept of a balanced network should be applied to as many payment and clearing systems as possible.

\section{Notation and Definitions}

We will use standard matrix and lattice algebra. The notation and basic definitions are based on work of Eisenberg and Noe \cite{SRinFN}. We use boldface to denote vector character and uppercase Latin letters for matrices and also for sets. $\mathcal{G}$ is reserved to indicate a graph, and $\mathcal{N} = \{1, 2, ... , n\} \subset \mathbb{N}$. For any two vectors $\mathbf{x}, \mathbf{y} \in \mathbb{R}^n$, define the lattice operations
\begin{align}\label{lattice-operations}
    \begin{split}
        \mathbf{x}^+ &:= (max[x_1,0], max[x_2,0],\ \cdots , max[x_n,0]) \\
        \mathbf{x}^- &:= (-x)^+ = (max[-x_1,0], max[-x_2,0],\ \cdots , max[-x_n,0]).
    \end{split}
\end{align}
Let $\mathbf{1}$ represent an $n$-dimensional vector all of whose components equal $\mathbf{1}$, i.e., $\mathbf{1} = (1,\ \cdots , 1)$. Similarly, $\mathbf{0}$ represents an $n$-dimensional vector all of whose components equal 0. Let $\left \| \cdot \right \| $ denote the $l^1$-norm on $n$. That is, 
\begin{equation}\label{l-norm}
    \left\| \mathbf{x} \right\| := \sum_{i=1}^{n} \left| x_i \right|.
\end{equation}
We use the following terms:
\begin{itemize}
    \item \textbf{Obligation network} - is a directed graph where the nodes\footnote{Usually referred to as vertices in graph theory.} represent firms and the edges represent the obligations. Parallel edges are allowed to represent multiple obligations between two firms.
    \item \textbf{Nominal liabilities matrix} - is a matrix representing total obligations or liabilities between firms. We will define special vectors to describe properties of the nominal liabilities matrix.
    \item \textbf{Payment system} - is constructed by adding special function nodes to the obligation network. These special nodes represent sources of funds and a store of value. They can have connections to all nodes in the obligation network, and the set of all connections for each special node will be described by a vector.
\end{itemize}

Let graph $\mathcal{G}$ represent the obligation network with $n$ nodes representing firms, $m$ edges representing obligations between firms, and the function $o(e)$ representing the value of a single obligation $e \in E$ between firm $v_i$ and firm $v_j$ (e.g.\ from a single invoice). The graph $\mathcal{G}$ may contain multiple edges from node $v_i$ to $v_j$. We use $(v_i,v_j) \subset E$ for the subset of $E$ that corresponds to all the edges between node $v_i$ and node $v_j$. These definitions are summarized formally as follows:
\begin{align}
    \mathcal{G} &= (V,E,s,t,o) 
        &&\mbox{directed graph of obligations between firms}
            \label{def1}\\
    V &= \{v_1, \cdots ,v_n\}
        &&\mbox{set of n nodes representing firms}
            \label{def2}\\
    E &= \{e_1, \cdots ,e_m\}
        &&\mbox{set of m edges representing individual obligations}
            \label{def3}\\
    e &\in E
        &&\mbox{individual edge from set $E$}
            \label{def4}\\
    s &\colon E \rightarrow V
        &&\mbox{assigns the source node to each edge}
            \label{def5}\\
    t &\colon E \rightarrow V
        &&\mbox{assigns the target node to each edge}
            \label{def6}\\
    o &\colon E \rightarrow \mathbb{R}
        &&\mbox{assigns the value of the obligation to each edge}
            \label{def7}
\end{align}

The \emph{nominal liability matrix} $L$ is a square $(n \times n)$ matrix each of whose entries is the sum of the obligations between two firms. Since companies do not invoice themselves $L$ has zeros on the diagonal. Each entry is given by
\begin{equation}\label{liability-matrix}
    L_{ij} = \sum _{e \in (v_i,v_j)} o(e).
\end{equation}
The sum of row $i$ of the nominal liability matrix represents the total debt of firm $i$ and the sum of column $j$ represents the total credit of firm $j$:
\begin{equation}\label{debt-and-credit}
    d_i = \sum _{j=1}^n L_{ij}, \qquad\qquad
    c_j = \sum _{i=1}^n L_{ij},
\end{equation}
where for the same firm $i = j$. Eqs.\ \eqref{debt-and-credit} provide the components of the system-wide \emph{credit vector} $\mathbf{c}$ and \emph{debt vector} $\mathbf{d}$. The difference between the credit and debt for each firm gives the obligation network's \emph{net position vector} $\mathbf{b}$:
\begin{align}\label{net-calculation}
    \begin{split}
        \mathbf{b}  &= \mathbf{c} - \mathbf{d} \\
                b_i &= c_i - d_i, \qquad
        d_i, c_i, b_i \in \mathbb{R},
        \quad i \in \mathcal{N}.
    \end{split}
\end{align}

\begin{definition}
    A vector $\mathbf{b}$ is called \emph{balanced} if the sum of its components equals 0:
    \begin{align}
        \sum_{i=1}^n b_i = 0.
    \end{align}
\end{definition}
\begin{theorem}
    Vector $\mathbf{b}$ representing the net positions of all firms is balanced.
\end{theorem}
\begin{proof}
    Every obligation that forms the liability matrix contributes towards the net position exactly twice, once as a credit and once as a debt. The sum of all credits is therefore equal to the sum of all debts and the sum of all the net positions equals zero:
    \begin{align}
        \begin{split}
            \sum_{j=1}^n c_j = \sum_{j=1}^n \sum_{i=1}^n L_{ij}
                &= \sum_{i=1}^n \sum_{j=1}^n L_{ij}
                 = \sum_{i=1}^n d_i \\
            \sum_{i=1}^n b_i = \sum_{i=1}^n (c_i - d_i)
                 = \sum_{i=1}^n c_i - \sum_{i=1}^n d_i
                &= \sum_{j=1}^n c_j - \sum_{i=1}^n d_i = 0.
        \end{split}
    \end{align}
\end{proof}
\begin{corollary}
    As a consequence of vector $\mathbf{b}$ being balanced, the sum of its positive vector components must be equal to the sum of the absolute value of its negative vector components:
    \begin{align}
          \left\| \mathbf{b}^+ \right\|
        = \left\| \mathbf{b}^- \right\|,
    \end{align}
    where $\mathbf{b}^+$ and $\mathbf{b}^-$ are calculated as defined in Eqs.\ \eqref{lattice-operations}.
\end{corollary}
A balanced net positions vector is important for the analysis of cashflow from external sources to the obligation network and vice versa. To visualize these definitions we use a small obligation network, as shown in Fig.\ \ref{fig:smalloblignet}, that consists of four nodes representing firms and arrows representing the individual obligations between them. The arrow labels represent the values of the obligations.

\begin{figure}[ht]
    \centering
    \includegraphics[width=8cm]{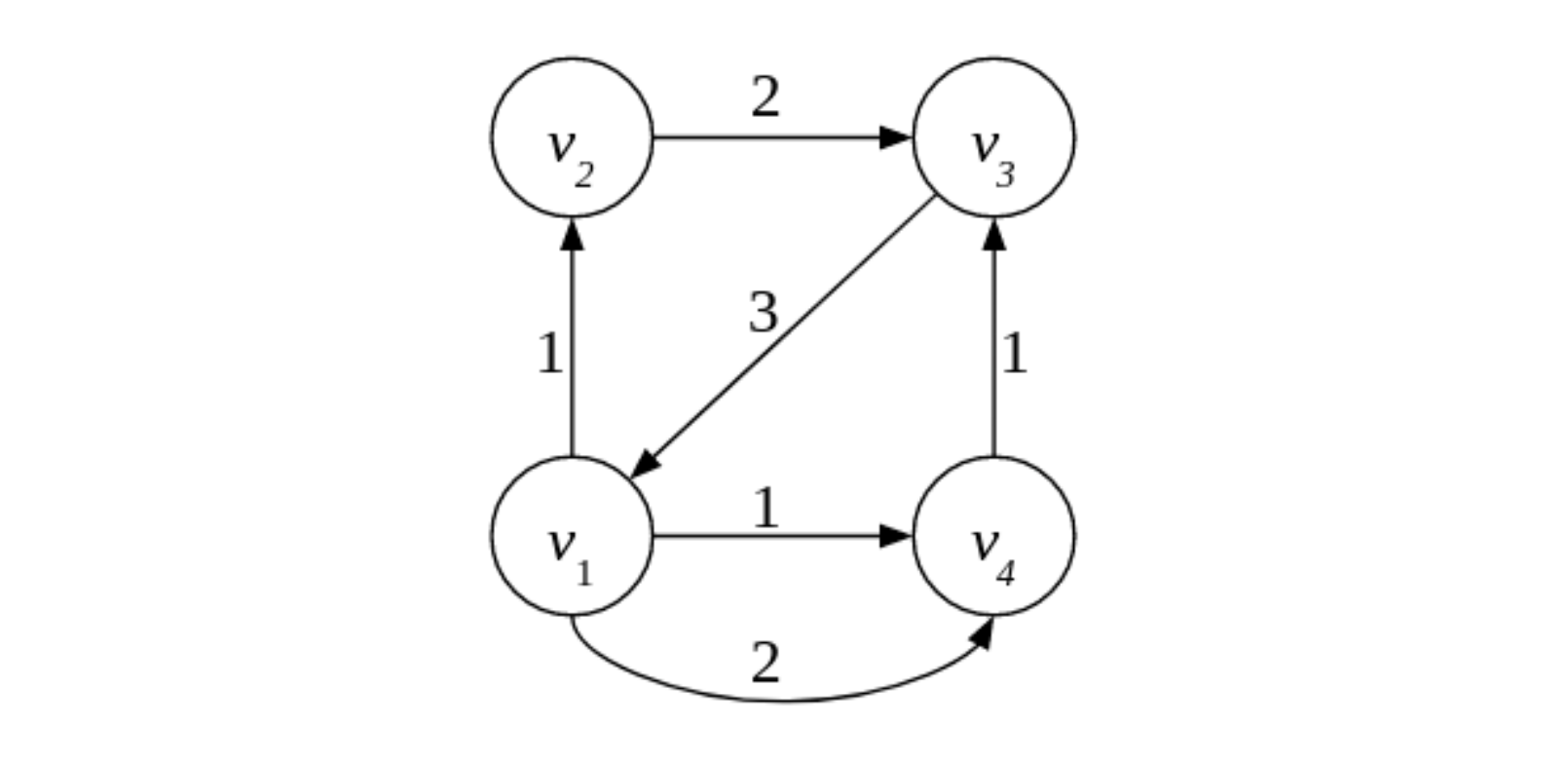}
    \caption{\small\textbf{Small obligation network}}
    \label{fig:smalloblignet}
\end{figure}
Eq.\ \ref{L matrix example} shows the corresponding nominal liabilities matrix $L$. Note that $L_{14}$ is the sum of the two obligations from Firm $1$ to Firm $4$. Eq.\ \ref{L matrix example} also shows the total credit and the total debt for each firm, as defined above.
\begin{align}\label{L matrix example}
    \begin{split}
        L = \begin{bmatrix} 0\; &1\;\,  &0\;\,  &3\,\\
                            0\; &0\;\,  &2\;\,  &0\,\\
                            3\; &0\;\,  &0\;\,  &0\,\\
                            0\; &0\;\,  &1\;\,  &0\,\end{bmatrix}&
                            \begin{array}{c}4\\2\\3\\1\end{array}
                            \begin{array}{c}d_1\\d_2\\d_3\\d_4\end{array}\\
        \begin{array}{cccc} 3\; &1\;    &3\;\,  &3\,\end{array}\;\\
        \begin{array}{cccc} c_1 &c_2    &c_3    &c_4\end{array}
    \end{split}
\end{align}
Vector $\mathbf{b}$ for this obligation network is calculated as
\begin{align}
    \begin{split}
        \mathbf{b} = \mathbf{c} - \mathbf{d} = (3,1,3,3)-(4,2,3,1)
            &= (-1,-1,0,2) \\
                \mathbf{b}^+ &= (0,0,0,2) \quad \Rightarrow \quad \left\| \mathbf{b}^+ \right\| = 2 \\
                \mathbf{b}^- &= (1,1,0,0) \quad \Rightarrow \quad \left\| \mathbf{b}^- \right\| = 2 \\
        \sum_{i=1}^n b_i = \left\| \mathbf{b}^+ \right\| - \left\| \mathbf{b}^- \right\| 
            &= 0.
    \end{split}
\end{align}

\section{Clearing All the Obligations in the Network}\label{ClearingAllObigations}

Our goal is to clear all the obligations in the obligation network. As shown in Fig.\ \ref{fig:oblignetfin}, to achieve this we introduce a special node $v_0$ that can act as a \emph{liquidity source} for all the cashflow towards the obligation network as well as a liquidity sink for all the cashflow from the obligation network. 

In practice $v_0$ can be a banking system where every firm in the network has a bank account. It can also be a complementary currency system or any other system with a store of value function. The cashflow is represented by an \emph{external cashflow vector} $\mathbf{f} \in \mathbb{R}^n$. When $f_i > 0$ the cashflow for firm $i$ is towards the obligation network, while when $f_i < 0$ its cashflow is from the network back to its bank account. By adding the cashflow vector we have created the payment system $(L,\mathbf{f})$.

\begin{figure}[ht]
    \centering
    \includegraphics[height=5cm]{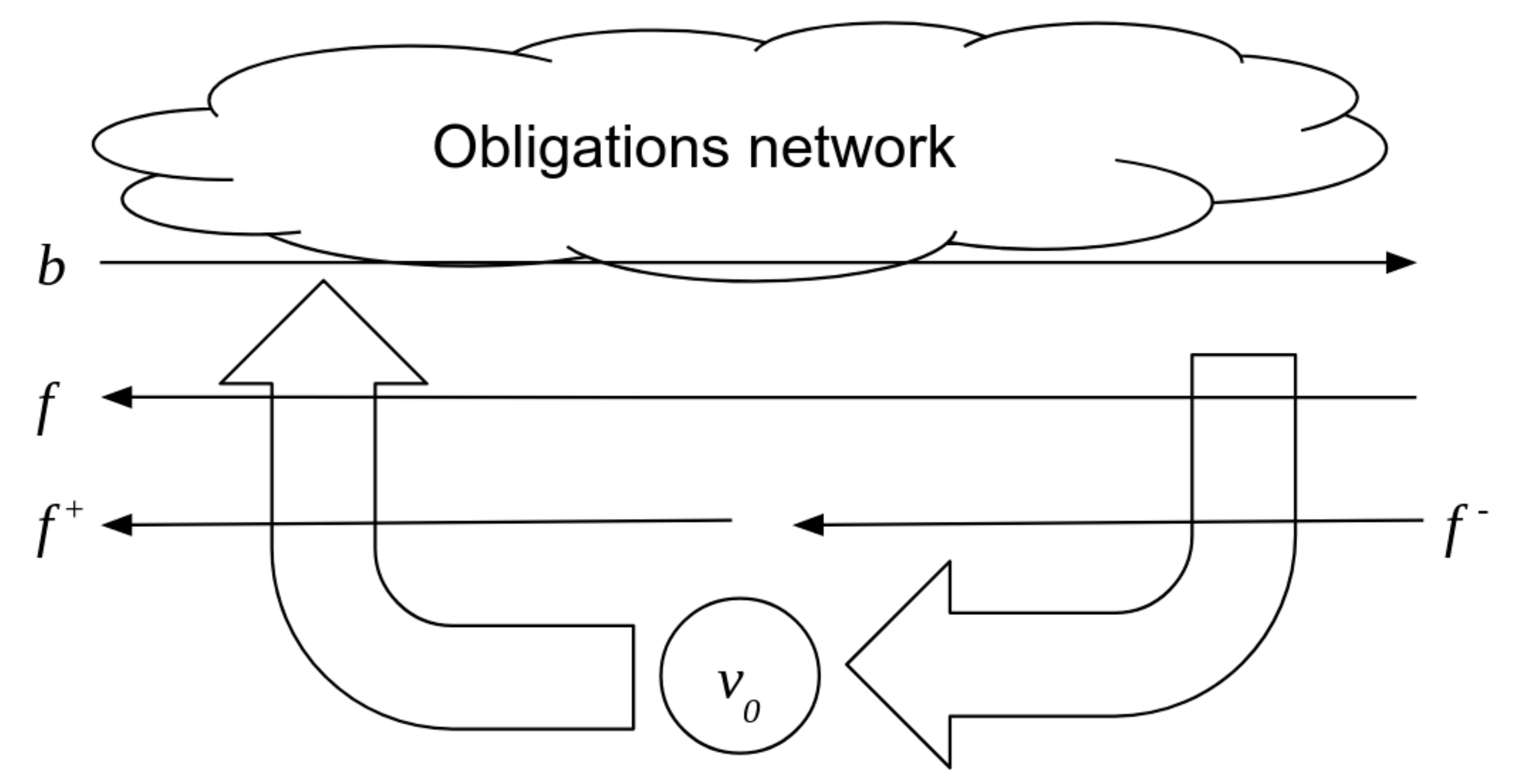}
    \caption{\small\textbf{Payment system: obligation network, source/sink of financing $v_0$, and vectors representing cashflows}}
\label{fig:oblignetfin}
\end{figure}

The cashflow available to firms from $v_0$ changes their net positions. If vector $\mathbf{b}$ represents the net positions of firms in the obligation network, let $\mathbf{b}^*$ represent the vector of firms' net positions in the payment system. The value of $\mathbf{b}^*$ is
\begin{equation}\label{b-star}
    \mathbf{b}^* = \mathbf{b} + \mathbf{f}.
\end{equation}
This simply states that the net position of every firm is increased by cashflow coming into the obligation network or decreased by cashflow going out of the obligation network. If our goal is to clear all the debt in the network, then -- assuming enough liquidity is available -- after our intervention the net position of every firm in the payment system has to be zero.  In such a scenario, the incoming cashflow is used to pay off the debts of all the firms with negative net positions, whereas the outgoing cashflow carries the cash into the bank accounts of the firms with positive (credit) net positions. Therefore,
\begin{equation}
\label{b-star-clear-all}
    \mathbf{b}^* = \mathbf{b} + \mathbf{f} = \mathbf{0} 
        \qquad \Rightarrow \qquad
            \mathbf{f} = -\mathbf{b}.
\end{equation}

Given an obligation network and the net positions $\mathbf{b}$ of its members, we now define:
\begin{definition}
\label{NIDdef}
    The \emph{Net Internal Debt} ($NID$) of the obligation network is the amount of cash needed by firms to discharge all the obligations in the network:
    \begin{equation}\label{NID}
        NID = \left|\left| \mathbf{b}^- \right|\right| .
    \end{equation}
\end{definition}
The payment system in Fig.\ \ref{fig:oblignetfin} relates to a real-life situation if we take that $v_0$ is a bank, complementary currency, or some other financial institution that can provide an account-holding function and/or that can serve as a source of liquidity. So the positive values of the cashflow vector $\mathbf{f}^+ = (-\mathbf{b})^+$ represent the payments from individual firms’ accounts at the financial institution, while the negative values of the cashflow vector (or the values of $\mathbf{f}^-$) represent the payments out of the network and into individual firms’ accounts. Since the vector $\mathbf{b}$ of the firms' net positions is balanced, the cashflow vector $\mathbf{f}$ is also balanced. So the total cashflow flowing into the network equals the cashflow out of the network:
\begin{align}\label{fplusEQfminus}
       \left|\left| \mathbf{f}^+ \right|\right|
     = \left|\left| \mathbf{f}^- \right|\right|.
\end{align}
\begin{definition}
    A payment system that can discharge all obligations in an obligation network is balanced.
\end{definition}
\begin{theorem}\label{BalacedPS}
    Payment system $(L,\mathbf{f})$ where $\mathbf{f}=-\mathbf{b}$ is balanced.
\end{theorem}
\begin{proof}
    A balanced payment system has to discharge all obligations in the obligation network. That is, for every firm or node in the obligation network the sum of all incoming and outgoing cashflows has to be 0. Therefore, the total credit minus the total debt as defined in \ref{debt-and-credit} plus the external financing has to be 0:
    \begin{align}
        c_i - d_i + f_i = b_i + f_i = b_i - b_i = 0 \qquad \forall i \in \mathcal{N}
    \end{align}
\end{proof}
\begin{corollary}\label{balanced payment follows flow conservation}
    Every balanced payment system satisfies the flow conservation constraint.
\end{corollary}
\begin{proof}
    The flow conservation constraint requires all flows into a node to be equal to all outflows from a node. For a balanced payment system this is true for all nodes in the obligation network, as proven in Theorem \ref{BalacedPS}. It is also true for the special node $v_0$ since the sum of all outgoing cashflow $\left\|\mathbf{f}^+ \right\|$ equals the sum of all incoming cashflow $\left\|\mathbf{f}^- \right\|$, as shown by Eq.\ \eqref{fplusEQfminus}.
\end{proof}

\subsection{An Obligation Chain}
To demonstrate the idea of a balanced cashflow vector that clears all obligations in an obligation network, let us observe a small network with four firms that contains only one chain, Fig \ref{fig:chain}. Firm $1$ represented by $v_1$ has an obligation to pay $1$ to company $v_2$, and so forth. The three obligations imply the presence of three edges: $\{e_1, e_2, e_3\}$.
\begin{figure}[ht]
    \centering
    \includegraphics[width=8cm]{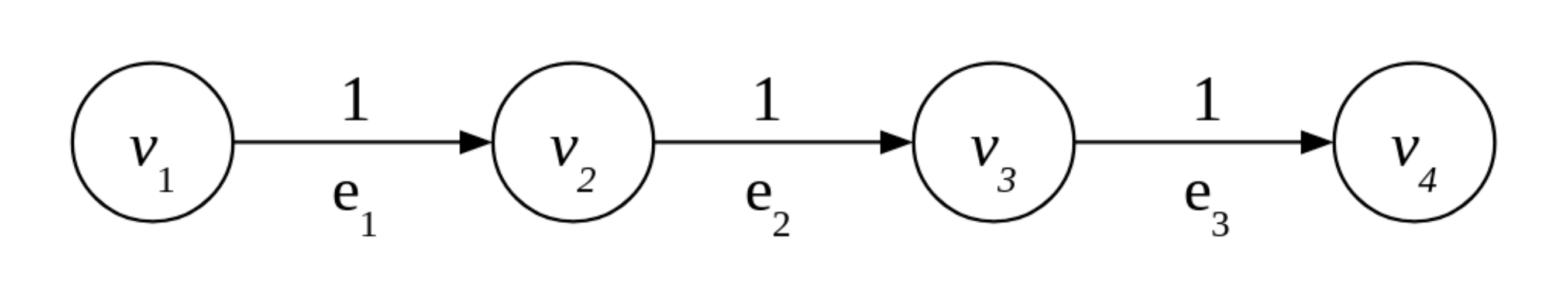}
    \caption{\small\textbf{A chain of obligations}}
\label{fig:chain}
\end{figure}
It is easy to see from the graph of the obligation network that, if Firm $1$ has access to one unit of account of liquid assets, all the firms in the chain can clear all their obligations, resulting in Firm $4$ having one unit of account more in their assets. Vector $\mathbf{b}$ for this obligation network is
\begin{equation}
\label{n-distribution}
    \begin{split}
        \mathbf{b} &= (-1,0,0,1) \\
        \mathbf{b}^- &= (1,0,0,0) \\
        NID = \left|\left| \mathbf{b}^- \right|\right| &= 1.
    \end{split}
\end{equation}
Therefore, the $NID$ or the amount of external liquidity needed to clear all obligations in this small obligation network containing only one chain is $1$.

As shown in Fig.\ \ref{fig:chainfin}, to create a payment system we have to add a new node $v_0$ representing the liquidity source with two edges: $e_5 = (v_0, v_1)$ with value $o_{(v_0,v_1)} = 1$, that represents the flow of cash into the obligation network, and $e_4 = (v_4, v_0)$ with value $o_{(v_4,v_0)} = 1$, that represents the flow of cash out of the obligation network. Therefore, clearing the obligation network leaves Firm $1$ with $1$ unit of account less in their bank account and Firm $4$ with $1$ unit of account more.

\begin{figure}[ht]
    \centering
    \includegraphics[width=8cm]{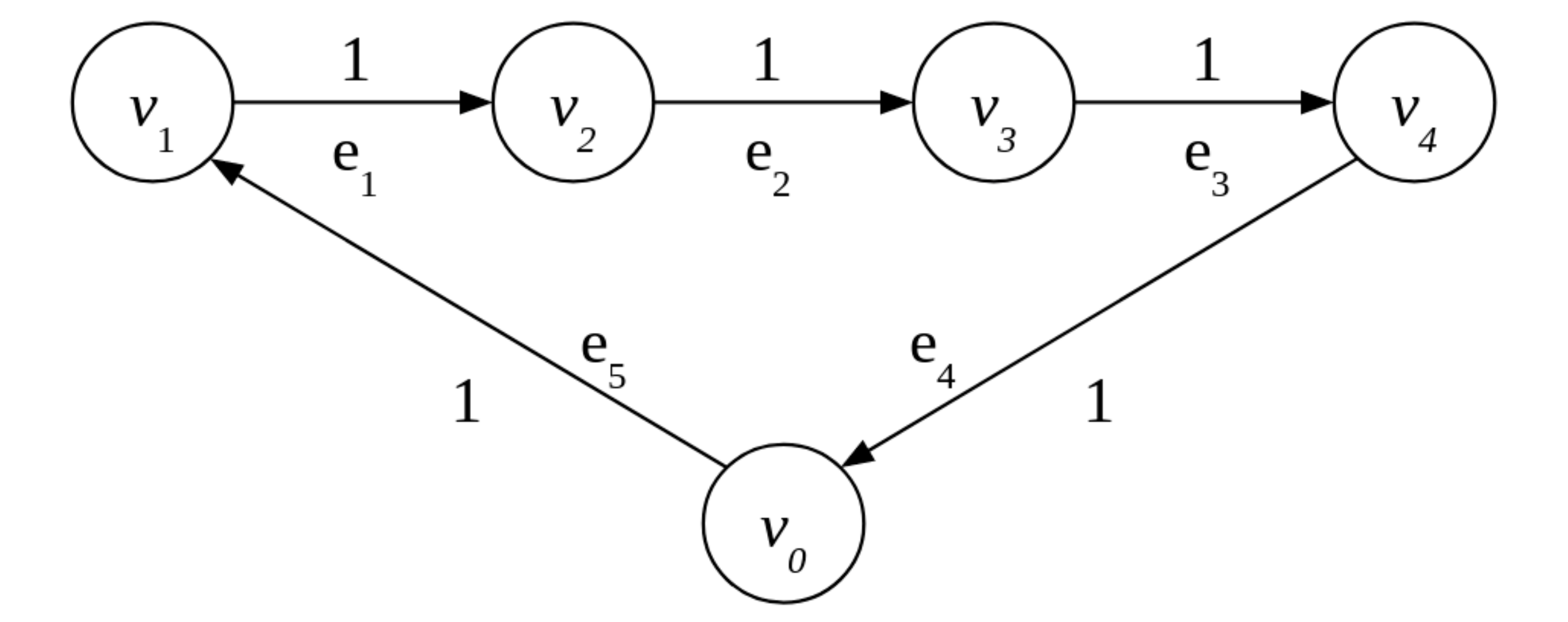}
    \caption{\small\textbf{A payment system with a chain of obligations}}
\label{fig:chainfin}
\end{figure}

Fig.\ \ref{fig:chainfin} also shows that providing liquidity is not just a problem concerning the total amount of liquid assets available, or $NID$, but also of their distribution. As shown in Eqs.\ \ref{n-distribution}, vector $\mathbf{b}^-$ shows the distribution of liquid assets needed to discharge all obligations in the obligation network. If we let the firms with just enough liquid assets to discharge all obligations act as independent actors, it will take three steps or three individual payments to discharge all obligations in the chain. Using a centralized queue with an LSM, on the other hand, will discharge all obligations simultaneously. This is an example of the time-saving property of LSMs.

\subsection{An Obligations Cycle}

Another interesting example of a small obligation network is a cycle, see Fig.\ \ref{fig:cycle}.
\begin{figure}[ht]
    \centering
    \includegraphics[width=8cm]{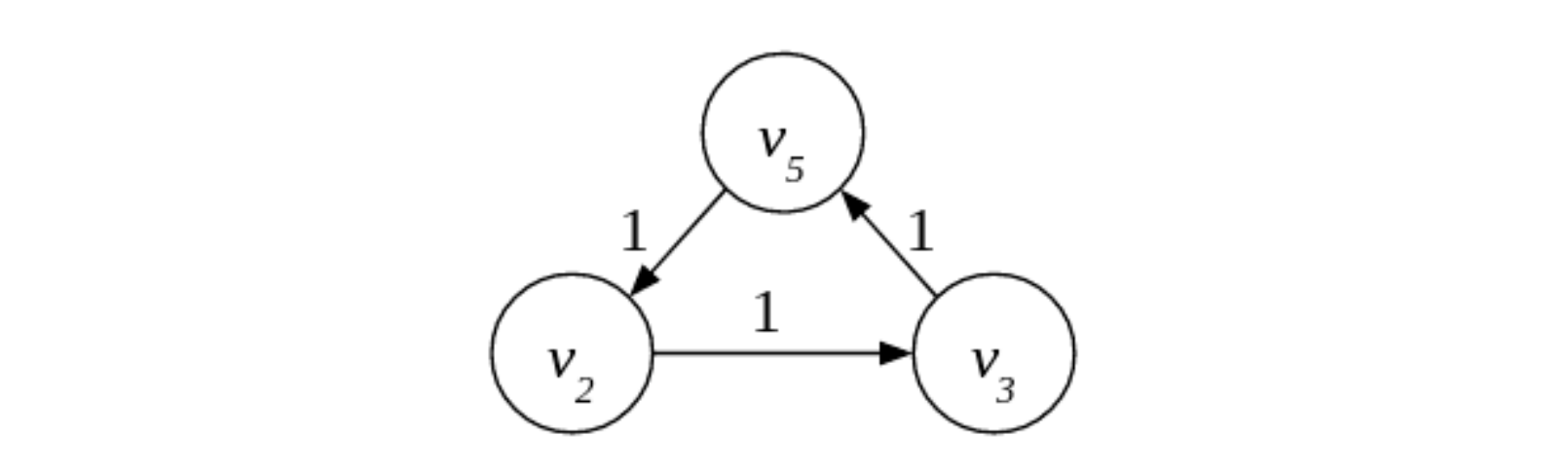}
    \caption{\small\textbf{A cycle}}
\label{fig:cycle}
\end{figure}

The $NID$ for this network is $0$. Therefore, all obligations in the cycle can be cleared without the use of external liquidity. In this special case where the vector $\mathbf{b} = \mathbf{0}$, we can say that the network is balanced. That also means that a cycle meets the flow conservation constraint at all nodes. In other words, for each firm $i$,
\begin{equation}
\label{cycle-flow-conservation}
    \sum _{j=1} ^n L_{ji} - \sum _{j=1} ^n L_{ij} = 0 \qquad \forall i \in \mathcal{N}.
\end{equation}
This flow conservation constraint equation can be written as a special case of the net positions calculation, Eq.\ \ref{net-calculation}. So the flow conservation constraint is met when all credits equal all debts for every firm in the cycle.

Although there is no need for external liquidity sources, such a simple cycle cannot be discharged if firms act as independent agents. Without the knowledge of the existence of such a cycle, the payments to discharge the obligations cannot be executed. To discharge all the obligations in a cycle without knowledge of its existence, at least one of the firms in the cycle has to use external liquidity to execute the first payment that then cascades around the cycle. Only with a centralized queue and LSM can we discharge all obligations in a cycle without the use of external liquidity sources. This is the liquidity-saving property of LSMs.

Cycles in obligation networks are the key to liquidity-saving. At this point it is worth noting that a cashflow that discharges obligations in the payment system is always a cycle. This cycle can form inside the obligation network as our example in Fig.\ \ref{fig:cycle}, or it can pass trough the special node $v_0$ as shown in Fig.\ \ref{fig:chainfin}. The flow conservation constraint is met in this case too.

\subsection{Small obligation network with a Chain and a Cycle}
Combining a chain and a cycle in a small obligation network we move closer to a real-life situation. Fig.\ \ref{fig:cyclechain} shows the union of the chain and cycle discussed above.
\begin{figure}[ht]
    \centering
    \includegraphics[width=8cm]{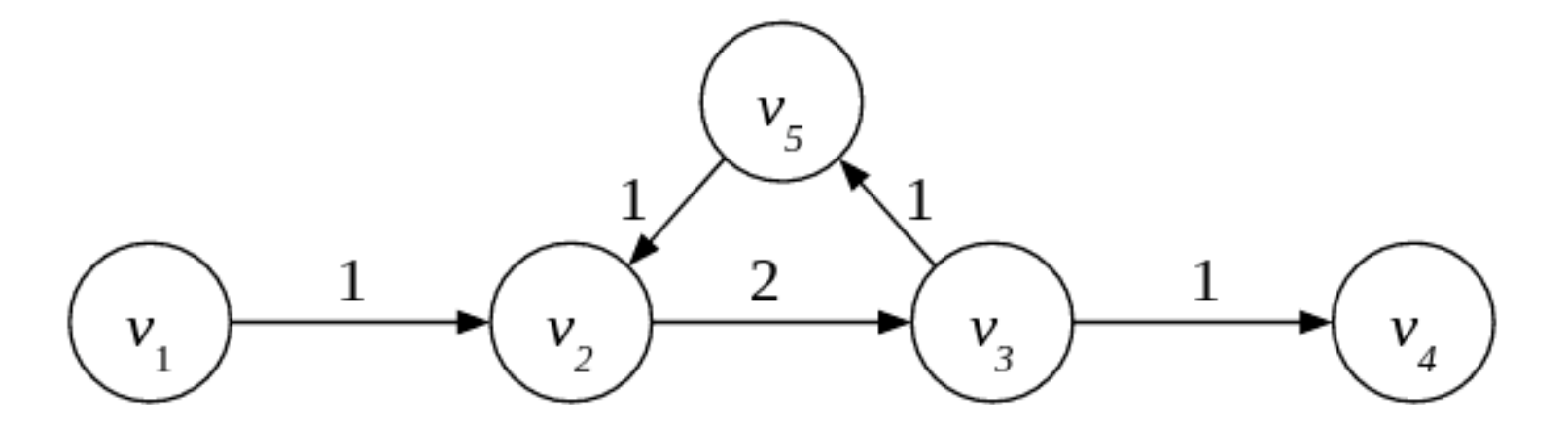}
    \caption{\small\textbf{An obligation network with a chain and a cycle}}
\label{fig:cyclechain}
\end{figure}
 The obligation network shown is obviously not balanced and needs external sources of liquidity to discharge all the obligations. The solution is similar to the chain example. The vector $\mathbf{b}$ for this obligation network is
\begin{equation}
    \begin{split}
        \mathbf{b}  & = (-1,0,0,1,0) \\
        \sum_{i=1}^n b_i & = 0 \implies \mbox{ \textbf{b} is balanced} \\
        \left \| \mathbf{b} \right \| & = 2 \implies \mbox{ obligation network is not balanced} \\
        \left \| \mathbf{b}^- \right \| & = 1 \implies NID = 1
    \end{split}
\end{equation}
By definition vector $\mathbf{b}$ is always balanced, but the obligation network usually is not. In our case $NID = 1$, so this obligation network needs an external liquidity source that can provide $1$ unit of account to discharge all the obligations.

The payment system containing chain and cycle is shown in Fig.\ \ref{fig:cyclechainfin}. Although there is enough liquidity in such a system to discharge all obligations, it cannot be done if members of the system act as independent agents. The cycle in the system prevents the smooth flow of cash. Firm $2$ cannot discharge its obligations even when it receives payment from Firm $1$. This creates a gridlock that can be resolved in several ways. One way is that Firm $2$ borrows from an external source. That implies the need for another edge from node $v_0$ to node $2$ with value $1$. The borrowed funds can be returned to $v_0$ as soon as the payment from Firm $5$ to Firm $2$ is executed.

\begin{figure}[ht]
    \centering
    \includegraphics[width=8cm]{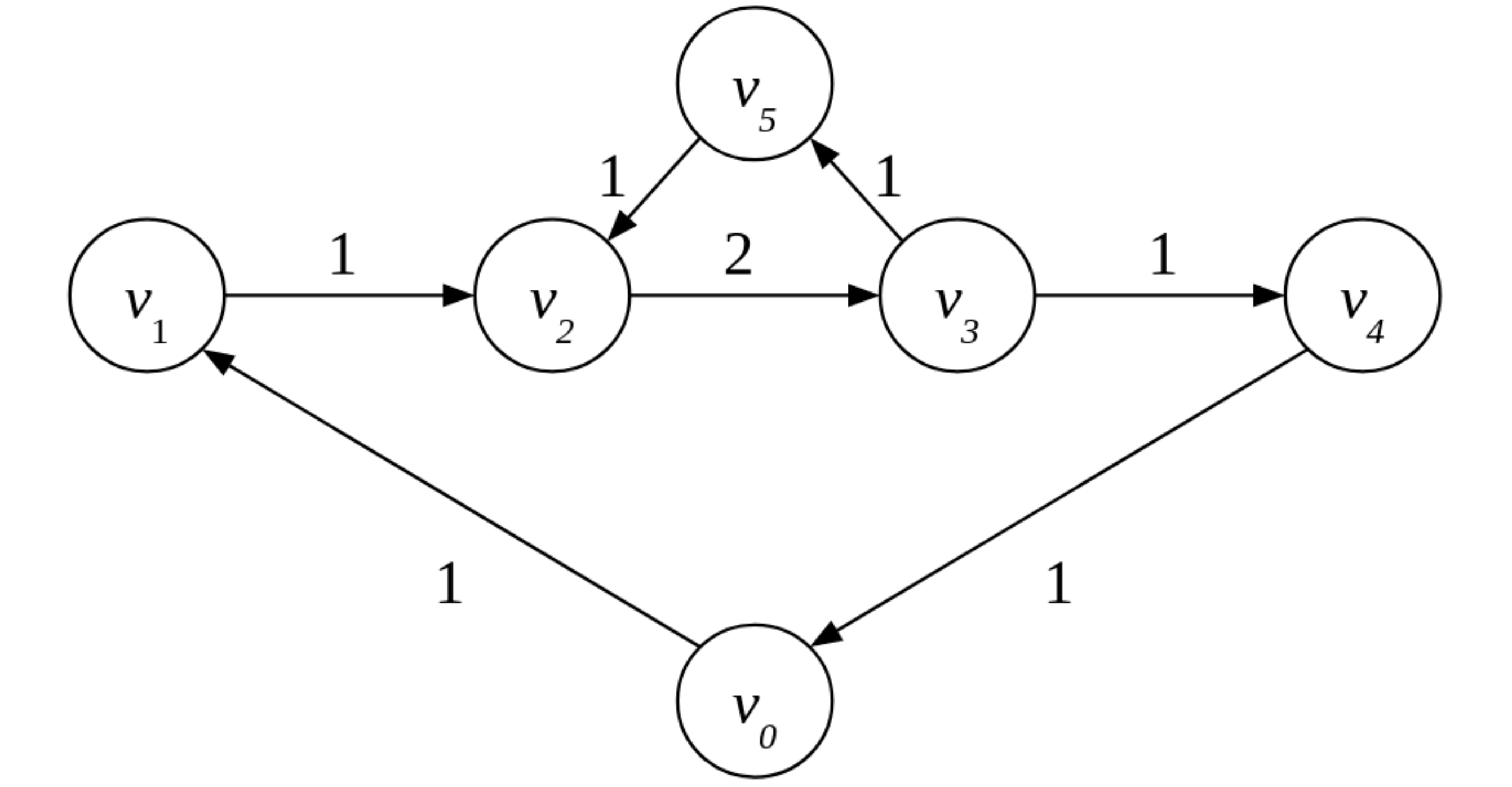}
    \caption{\small\textbf{Payment system with a chain and a cycle}}
\label{fig:cyclechainfin}
\end{figure}

This scenario is depicted in Fig.\ \ref{fig:gridresA}. Another way to resolve the gridlock is for any other firm in the cycle to borrow from an external source, which would require new edges from node $v_0$. The third option, which still assumes that 1 unit arrives at $v_2$ from $v_1$, is that Firms $2$ and $3$ agree on the partial discharge of the obligation between them. In this case, the partial payment of 1 unit of account from Firm $2$ enables Firm $3$ to discharge one of its obligations. If they decide to discharge the obligation to Firm $5$, the cycle will be discharged in full. This removes the gridlock situation created by the cycle. The flow of 1 unit from $v_0$ trough the obligation network is therefore unobstructed and all the remaining obligations can be cleared. If Firm $3$ decides to discharge the obligation towards Firm $4$ before Firm $5$ we are back to gridlock.

\begin{figure}[ht]
    \centering
    \includegraphics[width=8cm]{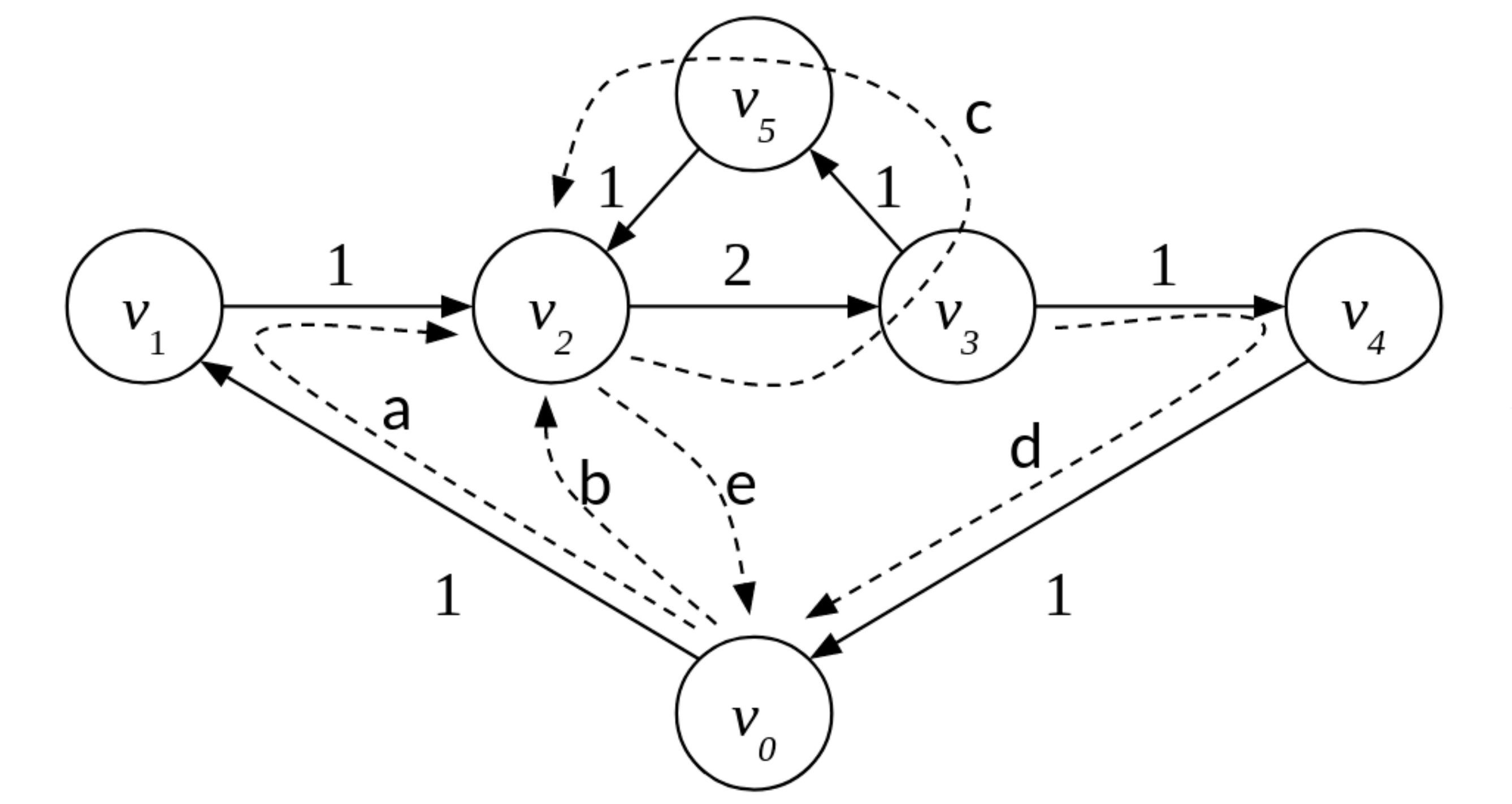}
    \caption{\small\textbf{Example of gridlock resolution scenario. Sequence of steps depicted with dashed arrows is marked with letters from ``a'' to ``e''}}
\label{fig:gridresA}
\end{figure}

Only putting the whole obligation network in a queue with an LSM will resolve the gridlock situation without the need for additional liquidity or special agreements among the firms. The solution is to identify the cycles that discharge the obligations simultaneously. Our example contains two cycles. The first, smaller cycle involving firms $(2,3,5,2)$ is located inside the obligation network. Obligations with value of $1$ can be discharged and the cycle can be removed from the obligation network without affecting the value of vector $\mathbf{b}$. Therefore, the $NID$ or the minimal requirement for external liquidity to discharge all the obligations in the obligation network remains the same. The situation in the payment system after removal of this cycle is a chain with the liquidity source, as shown in Fig.\ \ref{fig:chainfin}. This chain with a liquidity source forms the second cycle that is discharged in full with the use of liquidity from the external source.

\section{General Formulation}

In the following sections we build on the basic definitions and examples of the previous sections to develop some deeper results with important applications to large-scale payment systems.

\subsection{A Cycle Is a Balanced Payment Subsystem}

A cycle $E_c = \{e_{c1}, e_{c2}, \dots e_{ck} \}$ in $\mathcal{G}$ is a closed sequence of $k$ edges that connects nodes $V_c = \{v_{c1}, v_{c2} \dots v_{ck}, v_{c1}\}$ consecutively, where $v_{c1} \dots v_{ck}$ are distinct.
\begin{theorem}
    A cycle $(V_c,E_c,p)$ where all weights of the edges are set to the minimum obligation in the cycle $p = min(o(e) | e \in E_c)$ is a \emph{balanced payment subsystem}.
\end{theorem}
\begin{proof}
    The net position $b_i$ of every node in such a system is equal to zero since all obligations are equal to $p$ and every node has exactly one incoming and one outgoing edge with the same value.
    \begin{align}
        b_i = p - p = 0 \qquad \forall i \in \{1 \dots k \}.
    \end{align}
    Therefore, such a balanced payment subsystem meets the flow conservation constraint of Eq.\ \eqref{cycle-flow-conservation} and thus does not need external sources of liquidity to clear all its obligations.
\end{proof}
\begin{theorem}
\label{substracting cycle}
    Subtracting a cycle from the obligation network does not change the vector of net positions $\mathbf{b}$.
\end{theorem}
\begin{proof}
    All the flows in a cycle are equal to $p$. When a cycle is subtracted from the obligation network both credit and debt positions of every firm included in a cycle decrease by the same amount $p$, leaving the net position vector $\mathbf{b}$ unchanged:
    \begin{equation}
    \label{b-unchanged}
         b_i =  c_i - p - (d_i - p) = c_i - d_i = b_i \qquad \forall v_i \in V_c.
    \end{equation}
\end{proof}
It is important to note that, while from the system's viewpoint the removal of a balanced payment subsystem does not change the need for external financing, from the individual firm's perspective it makes a big difference, since with the removal of a balanced payment subsystem the corresponding gridlock situation is cleared and there is no need for external financing to resolve it. This reduction of the need for external financing can be observed as a reduction of the obligations in the payment system.

\subsection{Finding the Maximum-Weight Set of Cycles}
\label{maxcycles}
The problem of cycle elimination from directed graphs to get an acyclic graph is a well-studied area in graph theory. An overview of cycles in a graph is provided by \cite{CycleBasis}, and a fast parallel algorithm for finding the cycles in a graph is proposed in \cite{CycleRemoval}. Here we develop our own method, starting with the concept of ``weight''.
\begin{definition}
    We define the \emph{weight of an obligation network} with set of edges $E$ as a function $ w\colon \mathcal{G} \rightarrow \mathbb{R}$ whose value is the sum of all the obligations in the network:
    \begin{align}
        w(\mathcal{G}) = \sum_{e \in E} o(e).
    \end{align}
    Similarly, a cycle $\mathcal{G}_c$ of length $k$ with all its obligations of value $p$ has weight:
    \begin{align}
        w(\mathcal{G}_c) = pk.
    \end{align}
\end{definition}

Removing a cycle $\mathcal{G}_c$ from the obligation network reduces its weight. The weight of the residual obligation network $\mathcal{G}_r$ is:
\begin{align}
    w(\mathcal{G}_r) = w(\mathcal{G}) - w(\mathcal{G}_c)
\end{align}
Therefore, reducing the need for external financing by the individual firms can be achieved by removing all the cycles from the obligation network, which is also equivalent to resolving all the gridlocks. To achieve this we need to solve the following problem.
\begin{problem}
\label{maxweightproblem}
    Find a sequence of cycles $\mathcal{G}_{ci}$ and residual obligation networks $\mathcal{G}_{ri}$ such that
    \begin{align}
        \sum_{i=1}^q w(\mathcal{G}_{ci})
            \text{ is maximum},
    \end{align}
    where
    \begin{align}
        \begin{split}
            \mathcal{G}_{c1}, \mathcal{G}_{r1}
                &\subset \mathcal{G} \\
            \mathcal{G}_{c2}, \mathcal{G}_{r2}
                & \subset \mathcal{G}_{r1} \\
                & \vdots \\
            \mathcal{G}_{cq}, \mathcal{G}_{rq}
                & \subset \mathcal{G}_{r(q-1)}
        \end{split}
    \end{align}
    and $\mathcal{G}_{c1}, \hdots , \mathcal{G}_{cq}$ are cycles.
\end{problem}
Sequential elimination of cycles from the obligation network will always lead to a residual network $\mathcal{G}_{rq}$ that is acyclic. The exact number of cycles that will be eliminated is not known upfront and depends on the methods used to find them. For example, every directed acyclic graph has a topological ordering, i.e.\ an ordering of the vertices such that the starting end-point of each edge occurs earlier in the ordering than the ending end-point of that edge. The ordering can be found in linear time using Kahn’s algorithm for Topological Sorting \cite{Kahn}. This would be a possible formal test. Alternatively, the cycle-finding algorithm usually has a ``cycle not found'' exit condition which is also an acyclic graph test. The algorithm repeats until a cycle can be found. When no cycles can be found any longer, whatever is left is acyclic.

\begin{definition}
\label{maxweightsetcycles}
    Given an obligation network $\mathcal{G}$, a \emph{maximum-weight set of cycles} $\{ \mathcal{G}_{c1}, \hdots , \mathcal{G}_{cq} \}$ is one of the solutions to Problem \ref{maxweightproblem}.\footnote{There are many possible maximum-weight sets of cycles but only one maximum weight.}
\end{definition}

It is known that there is always a way to eliminate all the cycles and that the solution is not unique. The problem is that the removal of one cycle can break other embedded cycles, so the solution depends on the order in which the cycles are found. This makes finding the maximum-weight set of cycles even harder. The solution is not to look for cycles at all but to use the concept of balanced payment system and minimum-cost flow instead.

\subsection{The Minimum-Cost Flow Problem}
The starting point of the method is a perfectly balanced payment system as described in Sec.\ \ref{ClearingAllObigations}. We have an obligation network $\mathcal{G}$ with associated nominal liabilities matrix $L$, net position vector $\mathbf{b}$, and external cashflow vector $\mathbf{f}=-\mathbf{b}$, forming a balanced payment system $(L,\mathbf{f})$. We know that the cashflow through this balanced payment system equals the $NID$ as defined in Eq. \eqref{NID}. Now we try to find a balanced payment system $(M,\mathbf{f})$, where the nominal liabilities \emph{minimum-cost maximum-flow matrix} $M$ represents the minimum-weight sub-network $\mathcal{G}_m$ of the obligation network $\mathcal{G}$.
To find such nominal liabilities matrix $M$ we have to:
\begin{itemize}
    \item Define a Grandsum function
    $\mu\colon \mathbb{R}^{n^2} \rightarrow \mathbb{R}$ which is the sum of all the elements of a given square $n\times n$ matrix. Looking for the minimum of the function $\mu(M)$ is equivalent to looking for the minimum-weight sub-network $\mathcal{G}_m$.
    \item Make sure that that payment system $(M,\mathbf{f})$ is balanced. Therefore the column sum or credit vector minus the row sum or debt vector of matrix $M$ has to equal matrix $L$'s net positions vector $\mathbf{b}$.
    \item Make sure we are not introducing edges between nodes in sub-network $\mathcal{G}_m$ that do not exist in the obligation network $\mathcal{G}$. Therefore, all matrix elements $M_{ij}$ must have a value between 0 and $L_{ij}$. 
\end{itemize}

We can pose this as the following optimization problem:

\begin{problem}
\label{MCF problem}
    Find the liability matrix $M$ of the obligation network $\mathcal{G}_m$ such that its Grandsum function $\mu$ is minimum:
    \begin{align}\label{mcf2}
        min\ \mu(M) =
            min \sum_{i=1}^n \sum_{j=1}^n M_{ij}
    \end{align}
    subject to the constraints:
    \begin{align}\label{mcf3}
        \sum_{j=1}^n M_{ij} - \sum_{j=1}^n M_{ji} 
            &= b_i \qquad 
                &&\forall i \in \mathcal{N} \\
        0 \leq M_{ij} 
            &\leq L_{ij} \qquad 
                &&\forall i,j \in \mathcal{N}.
    \end{align}
\end{problem}

The reason we need to find $M$ is that it is the solution to the standard Minimum-Cost Flow (MCF) problem as defined in graph theory. The solution to the MCF problem equals all flows in a cycle from the liquidity source $v_0$, represented by vector $\mathbf{b}^-$, through the obligation network back to the liquidity source, represented by vector $\mathbf{b}^+$. We are looking for the shortest paths that can carry the NID through the obligation network.

Consistently with standard graph theory, we define a new node $s$ as a source and connect it to nodes representing firms with negative net positions $\mathbf{b}^-$; and we define node $t$, connecting it to nodes with positive net positions $\mathbf{b}^+$. The flow is set to the flow that clears all the obligations, i.e.\ $NID$ or $\left \| \mathbf{b}^- \right \|$. In the standard MCF solution the cost of the flows through different edges can vary. In our case the cost of the flow through different edges is the same, since we do not want to prioritize any specific flow or firm. So using the standard MCF solution all costs of a flow through an edge are set to 1.

Any minimum-cost flow algorithm will find a set of chains that can carry the max-flow $NID$ through the obligations at minimum cost. There are many known algorithms to solve the minimum-cost flow problem, e.g.\ see \cite{Kiraly} for an overview. A polynomial-time algorithm was proposed by \cite{Orlin}. The solution is not unique, but the value and the cost of the flow through the edges of the set of minimum-cost flows are always the same.

\begin{theorem}
\label{Tetris-theorem}
    Subtracting the minimum-cost max-flow flow solution $M$ from the nominal liabilities matrix $L$ leaves a balanced payment subsystem $(T,\mathbf{0})$ that requires no external liquidity source to clear the obligations:
    \begin{align}\label{Tetris}
        T = L - M.
    \end{align}
\end{theorem}
\begin{proof}
    This means that all the edges in the remaining nominal liabilities matrix $T$ are part of a cycle. The $T$ must be balanced, so the vector $\mathbf{b}$ for the matrix $T$ must be $\mathbf{0}$.
    \begin{align}\label{T_balanced}
        \sum _{j=1} ^n T_{ij} - \sum _{j=1} ^n T_{ji} = 0 \qquad \forall i \in \mathcal{N}
    \end{align}
    We can show that this is always true since $L$ and $M$ have the same vector $\mathbf{b}$.
    \begin{align}
        \begin{split}
            T_{ij} & = L_{ij} - M_{ij} \\
            \sum _{j=1} ^n T_{ij} - \sum _{j=1} ^n T_{ji} & = 0 \qquad \forall i \in \mathcal{N} \\
            \sum _{j=1} ^n (L_{ij} - M_{ij}) - \sum _{j=1} ^n (L_{ji} -M_{ji}) & = 0 \qquad \forall i \in \mathcal{N} \\
            \sum _{j=1} ^n (L_{ij} - L_{ji}) - \sum _{j=1} ^n (M_{ij} -M_{ji}) & = 0 \qquad \forall i \in \mathcal{N} \\
            b_i - b_i & = 0 \qquad \forall i \in \mathcal{N}
        \end{split}
    \end{align}
    This proves that the $T$ is composed of cycles only.
\end{proof}
\begin{corollary}
    $T$ is a \emph{maximum-weight} balanced payment subsystem.
\end{corollary}
\begin{proof}
    Since we subtracted the minimum value of chains $M$ from $L$, the remaining obligation network $T$ consists of the maximum value of cycles. Therefore, we have a maximum-weight balanced payment subsystem $T$.
\end{proof}

\section{Using Balanced Payment Subsystems in the Trade Credit Market}

In normal business situations we seldom have enough cash available to clear all our obligations. Therefore, we have to adjust our model to reflect the scarcity of liquidity.

The trade credit market is an interesting example since there are no liquidity sources at all. We can look at it as a payment system where the external financing vector $\mathbf{f}$ equals $\mathbf{0}$. So we have a payment system $(L,\mathbf{0})$ Still, it is possible to discharge the obligations bilaterally or multilaterally by applying the balanced payment subsystem idea. 

\begin{theorem}
    Subtracting a \emph{balanced payment subsystem} from the payment system does not change the net position vector $\mathbf{b}$.
\end{theorem}
\begin{proof}
    We will follow the similar path as in proof of theorem \ref{substracting cycle}. We have a nominal liabilities matrix $L$ with a net positions vector $\mathbf{b}$, credit vector $\mathbf{c}$ and debt vector $\mathbf{D}$. The balanced payment subsystem $(T,\mathbf{0})$ satisfies the flow conservation constraint as shown in corollary \ref{balanced payment follows flow conservation}. That means the cashflow into each node of the balanced payment subsystem equals the outflow. We can define a clearing vector $\mathbf{p}$ such that $\mathbf{p}_i$ stands for the flow into or flow out of node $i$. Now we subtract the balanced payment subsystem.
    \begin{align}
        L - T = M
    \end{align}
    and show that net positions vector of resulting nominal liabilities matrix $M$ equals the net position vector of nominal liabilities matrix $L$.
    \begin{align}
        b_i = (c_i - p_i) - (d_i - p_i) = c_i - d_i = b_i \qquad \forall i \in \mathcal{N}
    \end{align}
\end{proof}

The method to discharge the maximum amount of obligations without using any liquidity can be summarized in these steps:
\begin{enumerate}
    \item Collect obligations to form an obligation network $\mathcal{G}$.
    \item Form a nominal liabilities matrix $L$ and a payment system $(L,\mathbf{0})$ without external financing.
    \item Find maximum-weight balanced payment subsystem $T$.
    \item Discharge the obligations in a balanced payment subsystem $(T,\mathbf{0})$ by sending set-off notices to all pertaining firms.
    \item Substract the balanced payment subsystem $L - T = M$.
    \item Leave the remaining obligations in the nominal obligations matrix $M$ to discharge using normal bank payment system.
\end{enumerate}

The methodology and its use are described in \cite{ScharaBric}. It is marketed under name Tetris Core Technologies (TETRIS) and has been used in Slovenia since 1991 in support of the trade credit market. We can call matrix $T$ a \emph{TETRIS solution}. Depending on the economic conditions TETRIS discharges between 1\% and 5\% of GDP per year in saved liquidity towards the clearing of trade credit obligations. This is an example of an LSM in the trade credit market with a significant contribution to national financial stability.

\section{Multiple Cashflow Sources}\label{various-sources}
The cashflow vector $\mathbf{f}$ can be a sum of several different vectors. Fig.\ \ref{fig:paysysfirmodraft} shows the payment system with structured sources of cashflow.

\begin{figure}[ht]
    \centering  
    \includegraphics[height=8cm]{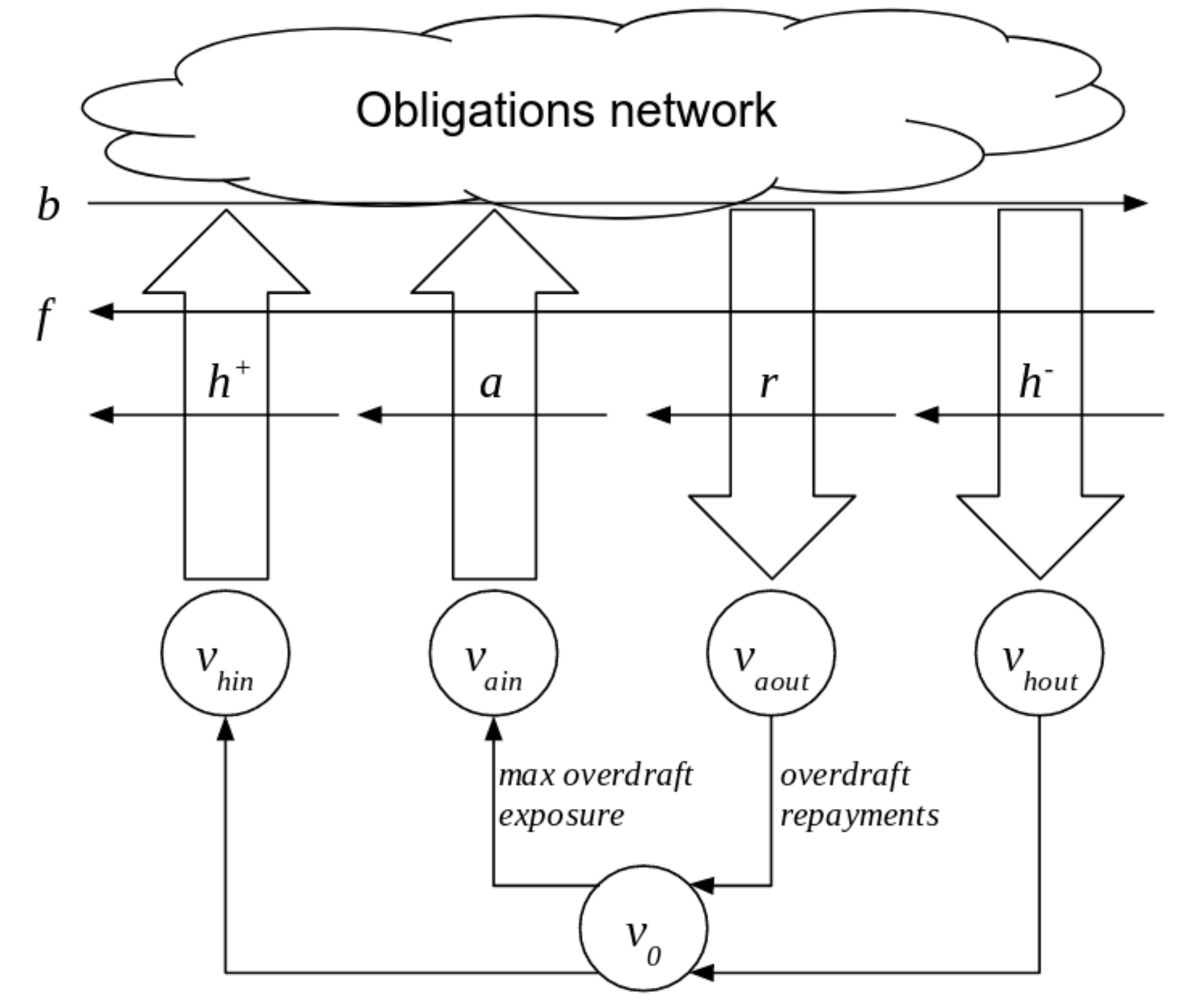}
    \caption{\small\textbf{A payment system with firms accounts and an overdraft facility}}
\label{fig:paysysfirmodraft}
\end{figure}

The first source is the account holding balance for each firm in the obligation network. This is the \emph{account holdings vector} $\mathbf{h}^+$ describing the available cashflow from node $v_{hin}$. The bank offers firms an overdraft facility. The maximum available overdraft for each firm is described by the \emph{available credit line vector} $\mathbf{a}$. The source of overdraft cashflow is node $v_{ain}$. The maximum exposure for the overdraft facility $a^{max} \in \mathbb{R}$ is set by the capacity of the edge connecting $v_0$ with $v_{ain}$. Node $v_{aout}$ is set to define the overdraft facility repayments cashflow. The \emph{repayments of credit limits vector} $\mathbf{r}$ represents the current overdraft taken by individual firms. This is the amount that has to be repaid to balance the individual firm's bank account. Node $v_{hout}$ is set to take the remaining cashflow $\mathbf{h}^-$ out of the obligation network back to accounts-holding node $v_0$. Let $\mathbf{a}^A$ be the approved maximum overdraft for firms. Then $\mathbf{a} = \mathbf{a}^A - \mathbf{r}$. When the distribution of the available external liquidity sources exceeds the values of vector $\mathbf{b}⁻$ at all points, a balanced payment system can be formed. When those conditions are met we can write:
\begin{equation}
    \begin{split}
        \mathbf{b}^- &\leq \mathbf{h}^+ + \mathbf{a} \\
        \left \| \mathbf{a} \right \| &\leq a^{max}  
    \end{split}
\end{equation}
Therefore there are enough external liquidity sourced in firms account holdings and overdraft facility to form an external financing vector $\mathbf{f}$ that satisfies the condition for a balanced payment system $(L,\mathbf{f})$
\begin{align}
    \mathbf{f} = -\mathbf{b}
\end{align}
If the conditions for a balanced payment system are not met, it makes sense to find a balanced payment subsystem to facilitate the discharge of as many obligations as possible.

\section{Optimizing the Use of Available Liquidity}
The optimal solution for a payment system described in Sec.\ \ref{various-sources} can be obtained by applying the idea from Theorem \ref{Tetris-theorem} to the payment system in its entirety. This ensures the maximum total obligation settlement amount with the available liquidity sources. Leaving the execution of payments to the discretion of the individual independent firms will yield a sub-optimal solution since firms do not have sufficient information about the payment system. To use the idea of removing the maximum weight of cycles from the payment system we have to transform the payment system into an extended nominal liabilities matrix, where liquidity sources become new nodes and the desired cashflows become the new liabilities. So we have additional nodes $v_0, v_{hin}, v_{hout}, v_{ain}$ and $v_{aout}$. Node $v_0$ is a debtor to nodes $v_{hin}$ and $v_{ain}$. This represents the total cashflow into the obligation network. The amount for the edge $(v_0,v_{ain})$ is set to the maximum overdraft facility exposure $a_{max}$. The amount for the edge $(v_0,v_{hin})$ is set to maximum potential cashflow into obligation network $\left \| \mathbf{b}^- \right \|$. Node $v_0$ is also a creditor to nodes $v_{hout}$ and $v_{aout}$. This represents the total cashflow out of the obligation network. The amount for the edge $(v_{hout}, v_0)$ is set to the maximum cashflow out of the network $\left \| \mathbf{b}^+ \right \|$. The amount for the edge $(v_{aout},v_0)$ is the current overdraft facility exposure $\left \| \mathbf{r} \right \|$.

Node $v_{hin}$ represents all the available cash in the individual firms’ accounts. For every firm $i$ with a positive cash balance, there is a connection between nodes $v_{hin}$ and $v_i$. The amount for this connection is the cash balance available. This setup ensures that the cashflow from the firms' bank accounts will never exceed the available cash in the individual accounts.

Similarly, we set up the available pre-approved bank account overdrafts. For firm $i$ with an overdraft approved there is a connection between nodes $v_{ain}$ and $v_i$. The amount for this connection is the amount of overdraft approved, $\mathbf{a}^A$, minus the overdraft already taken, $\mathbf{r}$. 
The repayments of the overdraft facility are set with the connection of nodes representing firms with overdraft taken to the node $v_{aout}$. The amounts for these connections are the individual firms’ overdrafts.

Let us denote such an extended nominal liability matrix by $L^*$. Let us apply the MCF algorithm to find an extended minimum-cost flow $M^*$ and then the extended maximum weight set of cycles $T^*$ by using the equation
\begin{equation}
    L^* - M^* = T^*
\end{equation}
\begin{theorem}
    The extended maximum weight set of cycles $T^*$ discharges the maximum amount of obligations in the obligation network with the available liquidity.
\end{theorem}
\begin{proof}
    The extended nominal liability matrix has no external liquidity sources. Therefore, new sources of liquidity are needed to discharge the obligations in $M^*$. Therefore, the cycles in $T^*$ used all available liquidity inside $L^*$ to discharge the maximum amount of obligations in the obligation network.
\end{proof}

\section{Conclusions}
The generalization of the payment system presented allows for the implementation of an LSM outside the interbank payment systems. The potential in trade credit markets is proven by 30 years of positive experience with trade credit clearing in Slovenia. New developments in e-invoicing and tax compliance create new opportunities to implement LSMs in the trade credit market.

There are new services developing that collect a huge amount of trade credit information that can be utilized to implement the idea of the balanced payment subsystem. New methods of information exchange using decentralized ledger technologies (DLT) call for the implementation of LSMs that provide solutions in environments where liquidity is not readily available.

Further opportunities for implementation are with complementary currencies that can enhance the implementation of an LSM through the use of the available mutual credit. It is also possible to design a payment system that discharges obligations issued in fiat currency by using a complementary currency as a source of liquidity. This way the benefits of mutual trust that are characteristic of complementary currency communities are transmitted to the wider society.

\bibliographystyle{unsrt}




\end{document}